\documentclass{article}
\usepackage{arxiv}
\usepackage[utf8]{inputenc} 
\usepackage[T1]{fontenc}    

\usepackage{cite}
\usepackage{algorithmic}
\usepackage[]{graphicx}
\usepackage[x11names]{xcolor}
\usepackage{textcomp}

\usepackage{amssymb, amsthm, amsmath, latexsym, amsfonts, mathrsfs, amsbsy,mathtools,relsize}
\usepackage{listings}
\usepackage{enumerate}
\usepackage{soul}
\usepackage{pgfplots}
\usepackage{tikz}
\usepackage{float}
\usepackage{dsfont}
\usepackage{verbatim}
\usepackage[ruled,vlined,linesnumbered]{algorithm2e}
\usepackage[colorlinks]{hyperref}
\usepackage{balance}

\hypersetup{
    citecolor=Firebrick4,
    linkcolor=Firebrick4,   
    urlcolor=Red3}
\usetikzlibrary{shapes,shapes.geometric,arrows,arrows.meta,fit,calc,positioning,automata,through,intersections}
\tikzset{diamond state/.style={draw,diamond}}

\newtheoremstyle{theoremdd}
  {\topsep}
  {\topsep}
  {\itshape}
  {0pt}
  {\bfseries}
  {.}
  { }
  {\thmname{#1}\thmnumber{ #2}\textnormal{\thmnote{ (#3)}}}

\theoremstyle{theoremdd}
\newtheorem{theorem}{Theorem}
\newtheorem{lemma}{Lemma}
\newtheorem{assumption}{Assumption}

\newtheorem{remark}{Remark}

\newcommand{\bm}[1]{\boldsymbol{#1}} 
\newcommand{\set}[1]{\mathcal{#1}} 
\newcommand{\ie}{\textit{i.e.,~}} 
\newcommand{\eg}{\textit{e.g.,~}} 
\newcommand{\inneighbor}[1]{\set{N}_{#1}^{\texttt{in}}}
\newcommand{\outneighbor}[1]{\set{N}_{#1}^{\texttt{out}}}
\newcommand{\indegree}[1]{d_{#1}^{\texttt{in}}}
\newcommand{\outdegree}[1]{d_{#1}^{\texttt{out}}}

\newcommand{\ouralgorithm}{{\fontsize{10.5pt}{10.5pt}\selectfont \textsc{pp-acdc}}}

\SetKwInput{KwRequire}{Require}

\newenvironment{list4}{
\begin{list}{$\bullet$}{%
    \setlength{\itemsep}{0.05cm}
    \setlength{\labelsep}{0.2cm}
    \setlength{\labelwidth}{0.3cm}
    \setlength{\parsep}{0in} 
    \setlength{\parskip}{0in}
    \setlength{\topsep}{0in} 
    \setlength{\partopsep}{0in}
    \setlength{\leftmargin}{0.18in}}}
{\end{list}}

\begin{document}
\title{Average Consensus with Dynamic Quantization Framing and Finite-Time Termination over Limited-Bandwidth Directed Networks}

\author{
  Evagoras Makridis\\
  Department of Electrical and Computer Engineering\\
  University of Cyprus\\
   Nicosia, Cyprus\\
  \texttt{makridis.evagoras@ucy.ac.cy} \\
   \And
  Gabriele Oliva \\
  Department of Engineering\\
  University Campus Bio-Medico of Rome\\
  Roma, Italy \\
  \texttt{g.oliva@unicampus.it} \\
  \And
  Apostolos I. Rikos \\
  Artificial Intelligence Thrust of the Information Hub\\
  The Hong Kong University of Science and Technology\\
  Guangzhou, China \\
  \texttt{apostolosr@hkust-gz.edu.cn} \\
  \And
  Themistoklis Charalambous\\
  Department of Electrical and Computer Engineering\\
  University of Cyprus\\
   Nicosia, Cyprus\\
  \texttt{charalambous.themistoklis@ucy.ac.cy}
}

\maketitle

\begin{abstract}
This paper proposes a deterministic distributed algorithm, referred to as {\ouralgorithm}, that achieves exact average consensus over possibly unbalanced directed graphs using only a fixed and \emph{a~priori} specified number of quantization bits. The method integrates Push–Pull (surplus) consensus dynamics with a dynamic quantization framing scheme combining zooming and midpoint shifting, enabling agents to preserve the true global average while progressively refining their quantization precision. We establish a rigorous convergence theory showing that {\ouralgorithm} achieves asymptotic (exact) average consensus on any strongly connected digraph under appropriately chosen quantization parameters. Moreover, we develop a fully distributed and synchronized finite-time termination mechanism, and we provide a formal proof on the detection of $\epsilon$-convergence to the average within a finite number of iterations. Numerical simulations corroborate the theoretical results and demonstrate that {\ouralgorithm} achieves reliable, communication-efficient, and precise average consensus even under very tight bit budgets, underscoring its suitability for large-scale and resource-constrained multi-agent systems operating over directed networks.
\end{abstract}


\section{Introduction}\label{sec:introduction}

Distributed consensus algorithms are essential components of modern applications that involve the cooperation of a group of networked agents (often referred to as \emph{multi-agent systems}) such as robotic networks, distributed power networks, and wireless sensor networks. Some of these applications include distributed optimization \cite{ carnevale2023distributed,maritan2024fully}, distributed estimation and control \cite{makridis2024fully, fioravanti2024distributed}, and distributed learning \cite{bastianello2024robust}. In such problems, the agents in a network aim at reaching agreement on a common decision by iteratively updating their values based on information received by their immediate neighbors via information exchange. When the agreed common value is the average of the initial values of all agents in the network, we say that the agents reach \emph{average consensus}. Although there exist several average consensus schemes considering network abnormalities (such as packet delays \cite{hadjicostis2013average,charalambous2024robustified} or drops \cite{hadjicostis2015robust,makridis2023harnessing}) and malfunctioning agents (e.g., curious, malicious, or faulty agents)~\cite{CHARALAMBOUS2024}, they require agents to be able to send and receive real values with infinite precision.

In multi-agent systems where the communication is often established over wireless channels, several constraints, such as limited bandwidth, power, and memory, arise. Such limitations often require that the information that is exchanged between agents is quantized. Quantization enables information compression such that a number in a continuous set (\ie real-valued number $x\in\mathbb{R}$), can be mapped to a value in a finite set. This mapping is done with a quantizer, denoted by $Q(\cdot)$. In this work, we define the quantizer given a positive integer number of bits $b<\infty$ available for quantization, the step-size $\Delta$ which defines the spacing between the quantization intervals, and the quantizer's symmetry point (or midpoint) $\sigma$ (which is not necessarily at $0$) as follows:
\begin{align}\label{eq:quantizer}
    Q\big(x,b,\Delta,\sigma\big) &= \begin{cases}
        \hfil \sigma + \Delta(2^{b-1} -1), & \text{if}~x>\overline{x},\\
        \hfil \sigma - \Delta(2^{b-1} -1), & \text{if}~x\leq \bar{x},\\
        \hfil \sigma + \Delta \bigl \lfloor \frac{x-\sigma}{\Delta} +\frac{1}{2} \bigr \rfloor, & \text{otherwise},
    \end{cases}
\end{align} 
where the upper and lower input quantization limits are given respectively by
\begin{subequations}
\begin{align}
    \overline{x} := \sigma + \Delta\big(2^{b-1} - 0.5\big),\\
    \underline{x} := \sigma - \Delta\big(2^{b-1} - 0.5\big).
\end{align}
\end{subequations}
An illustrative example of the input-to-output relation of a uniform quantizer with $b=3$, $\Delta=1$, and $\sigma=0$, is shown in Fig.~\ref{fig:quantizer}.
\begin{figure}[t]
    \centering
    \includegraphics[width=0.4\linewidth]{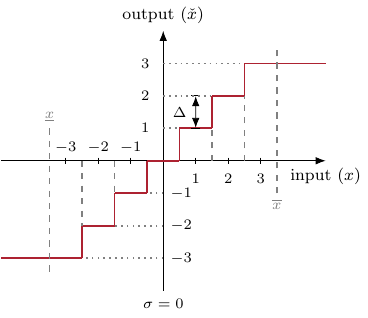}
    \caption{A uniform quantizer of $b=3$ bits with $\sigma=0$ and $\Delta=1$.}
\end{figure}
Although quantization reduces the precision of the original information, it enables mapping to a finite set of discrete values that are often representable using bits. This mapping is necessary for transmissions over channels with limited bandwidth and transmitters with power limitations. Bandwidth limits the number of symbols sent per second, while power constraints limit the number of bits per symbol that the modulation scheme can reliably encode. Increasing the number of bits per symbol requires higher transmission power to maintain reliable decoding at the receiver in the presence of noise. Consequently, quantization encodes real-valued information into sequence of bits aligned with the modulation scheme and constrained by the channel capacity.

Communicating quantized values in multi-agent systems has attracted significant interest due to its essential role in overcoming practical constraints like bandwidth and power limitations \cite{liu2024distributed,lin2024event,yan2025layered,doostmohammadian2025log,feng2025quantized}. However, in the context of distributed average consensus, quantized communication can negatively impact both the accuracy and overall performance of the system. Specifically, it affects the system behavior as follows. First, saturation occurs when the signal exceeds the quantization range, leading to large errors that can destabilize the system, if the consensus algorithm was originally designed for a non-quantized system. Second, as the system state gets closer to its target value, performance degradation arises because higher precision is required. Finally, each time an agent quantizes a value before sending it to its neighbors, some information is lost due to rounding and discretization. Over time, these small errors can accumulate, leading to biased consensus, \ie agents converge to a value that is far from the true average. These errors are even amplified when agents interact over a directed network, where the information flow is not balanced, since the errors may propagate asymmetrically within the network. Since this work focuses on average consensus in directed graphs, the discussion of related work is limited to literature specifically addressing possibly unbalanced directed graphs.

To minimize quantization errors, agents can use a dynamic quantization strategy, such as zooming in / zooming out (ZIZO) \cite{Brockett2000TAC}, where the quantization step size changes based on the range of values being exchanged. However, since the links are not bidirectional, it is difficult to synchronize the nodes to simultaneously update the size of the quantization step. In the same spirit, the authors in \cite{frasca2009average}, proposed a distributed algorithm in which the exact average is computed asymptotically. Their method works for balanced directed graphs, where each node maintains real-valued states but transmits quantized values, adapting the quantization step as needed to refine the estimates.
The majority of algorithms designed for quantized average consensus for unbalanced directed graphs achieve convergence in a probabilistic manner (see, \eg~\cite{2022:Rikos_Hadj_Johan} and references therein), meaning that nodes reach consensus on the quantized average with probability one or through other probabilistic guarantees. Nevertheless, either residual quantization errors persist indefinitely or the average consensus is approximate, the convergence is slow and depends on the network size. With only a few exceptions~\cite{Rikos2021TAC,Zambieri2009ECC,lee2020finite,song2022compressed,rikos2024finite}, the development of deterministic distributed strategies for achieving quantized average consensus in unbalanced directed graphs remains largely unexplored. Baldan and Zambieri~\cite{Zambieri2009ECC} proposed distributed algorithms that achieve average consensus with quantized transmissions and the algorithm's parameters can be chosen in a distributed fashion without knowing the number of agents composing the network. Since they only adopt the ZIZO strategy without any adaptation on the symmetry point $\sigma$ of the quantizer, there is a minimum number of bits required for the algorithm to converge which depends on the parameter choice. Rikos \emph{et al.}~\cite{rikos2024finite} proposed a distributed algorithm that is able to calculate the exact average of the initially quantized values (in the form of a quantized fraction) in a deterministic fashion after a finite number of iterations. However, the average consensus is approximate (due to the initial quantization), the number of bits for each message is not necessarily bounded by a certain value, and the convergence is slow and depends on the network size.


In this paper, we address the challenges of achieving exact average consensus over unbalanced directed graphs under limited communication bandwidth. Motivated by the limitations of existing probabilistic and approximate methods, we propose a deterministic distributed algorithm with dynamic quantization framing (\ie zooming and midpoint shifting) and surplus-consensus dynamics. Specifically, we:
\begin{list4}
\item Develop a distributed average consensus algorithm with dynamic quantization, allowing agents to reach the exact average with a fixed and finite number of bits (Section~\ref{sec:algo}). The algorithm adopts a surplus-consensus structure suited for unbalanced digraphs and avoids the division operations used in ratio-consensus methods \cite{SYS-016}, improving robustness to quantization errors.
\item Design a dynamic quantization framing mechanism, including zooming and midpoint shifting, directly into the consensus updates. This ensures that the consensus error asymptotically vanishes while maintaining deterministic guarantees under a fixed and finite number of communication bits.
\item Introduce a distributed finite-time termination mechanism that enables all agents to synchronously detect when their states are within an $\epsilon$-neighborhood of the average consensus value, so that they terminate their iterations in a coordinated and deterministic manner.
\end{list4}
These contributions are supported by rigorous theoretical analysis (Section~\ref{sec:algo}) and validated through numerical experiments (Section~\ref{sec:simulations}), emphasizing the practicality of our approach for resource-constrained directed networks.

\section{Preliminaries}\label{sec:background}

\subsection{Mathematical Notation}
The sets of real, integer, and natural numbers are denoted as $\mathbb{R}$, $\mathbb{Z}$, and $\mathbb{N}$, respectively. The set of nonnegative real (integer) numbers is denoted as $\mathbb{R}_{\geq 0}$ ($\mathbb{Z}_{\geq0}$). The nonnegative orthant of the $n$-dimensional real space $\mathbb{R}^n$ is denoted as $\mathbb{R}_{\geq 0}^{n}$. 
Matrices are denoted by capital letters, and vectors by small letters. The transpose of a matrix $A$ and a vector $x$ are denoted as $A^\top$, $x^\top$, respectively. The all-ones and all-zeros vectors are denoted by $\mathbf{1}$ and $\mathbf{0}$, respectively, with their dimensions being inferred from the context. The floor function of a real number $x$ is denoted as $\lfloor x \rfloor = \max \{ b \in \mathbb{Z} \mid b \le x \}$, and is defined as the greatest integer less than or equal to $x$. 

\subsection{Network Model}
Consider a group of $n>1$ agents communicating over a time-invariant directed network. The interconnection topology of the communication network is modeled by a digraph $\set{G}=(\set{V}, \set{E})$. Each agent $v_j$ is included in the set of digraph nodes $\set{V}=\{v_1, \cdots, v_n\}$, whose cardinality is $n=\left|\set{V}\right|$. The interactions between agents are included in the set of digraph edges $\set{E} \subseteq \set{V} \times \set{V}$. The total number of edges in the network is denoted by $m=\left|\set{E}\right|$. A directed edge $\varepsilon_{ji} \triangleq (v_j, v_i) \in \set{E}$ indicates that node $v_j$ receives information from node $v_i$. The nodes that transmit information to node $v_j$ directly are called in-neighbors of node $v_j$, and belong to the set $\inneighbor{j}=\{v_i \in \set{V} \mid \varepsilon_{ji} \in \set{E}\}$. The number of nodes in the in-neighborhood set is called in-degree and is denoted by $\indegree{j} = \left|\inneighbor{j}\right|$. The nodes that receive information from node $v_j$ directly are called out-neighbors of node $v_j$, and belong to the set $\outneighbor{j}=\{v_l \in \set{V} \mid \varepsilon_{lj} \in \set{E}\}$. The number of nodes in the out-neighborhood set is called out-degree and is denoted by $\outdegree{j}= \left|\outneighbor{j}\right|$. Each node $v_j \in \set{V}$ has immediate access to its own local state, and thus we assume that the corresponding self-loop is available $\varepsilon_{jj} \in \set{E}$, although it is not included in the nodes' out-neighborhood and in-neighborhood. 
A directed path from $v_i$ to $v_l$ with a length of $t$ exists if a sequence of nodes $i \equiv l_0,l_1, \dots, l_t \equiv l$ can be found, satisfying $(l_{\tau+1},l_{\tau}) \in \mathcal{E}$ for $ \tau = 0, 1, \dots , t-1$. In $\mathcal{G}$ a node $v_i$ is reachable from a node $v_j$ if there exists a path from $v_j$ to $v_i$ which respects the direction of the edges. The digraph $\mathcal{G}$ is said to be strongly connected if every node is reachable from every other node.

\subsection{Distributed Average Consensus Problem}

Let us assume that at each time instant $k\in\mathbb{Z}_{\geq0}$ each node $v_j \in \mathcal{V}$ maintains a scalar\footnote{The results can be trivially extended to the case where agents maintain vectorial states. In this work we consider scalar states for simplicity of exposition.} state $x_{j,k} \in \mathbb{R}$. In the distributed average consensus problem, the goal of the agents is to reach consensus to a value equal to the average of their initial states, which is defined as
\begin{align}\label{eq:ac_problem}
    x_{\text{ave}} := \dfrac{1}{n} \sum_{j=1}^{n} x_{j,0}.
\end{align}
Due to the absence of global knowledge at each agent, agents are required to execute an iterative distributed algorithm to eventually converge to the average consensus value, by means of local communication and computation. In particular, agents update their local states using information received from their immediate neighboring agents, through communication channels. 

In practice, however, the exchange of information is often restricted to be unidirectional (instead of bidirectional) as a consequence of diverse transmission power, interference levels, and communication ranges. Thus, although an agent $v_j$ may receive information from agent $v_i$, that does not necessarily imply that $v_j$ can also send information back to agent $v_i$. Consequently, the network topology is best modeled by a
directed graph (digraph).

\subsection{Average Consensus in Digraphs using Surplus Consensus}

A linear algorithm for achieving average consensus over directed and strongly connected networks was proposed in~\cite{cai2012average}. Each agent $v_j \in \set{V}$ maintains a local state $x_{j,k} \in \mathbb{R}$ and a surplus variable $s_{j,k} \in \mathbb{R}$, initialized as $x_{j,0} \in \mathbb{R}$, $s_{j,0} = 0$. At each iteration $k$, agent $v_j$ sends $x_{j,k}$ and $c_{lj}s_{j,k}$ to its out-neighbors $v_l \in \outneighbor{j}$, where the weights $c_{lj}$ are defined as:
\begin{align}\label{eq:c-weights}
    c_{lj}=\begin{cases}
    \frac{1}{1+\outdegree{j}}, & \text{if } v_l \in \outneighbor{j} \lor l = j,\\
    0, & \text{otherwise}.
    \end{cases}
\end{align}
Similarly, agent $v_j$ receives $x_{i,k}$ and $c_{ji}s_{i,k}$ from its in-neighbors $v_i \in \inneighbor{j}$, and updates its variables as:
\begin{subequations}\label{eq:ppac}
\begin{align}
x_{j,k+1} &= r_{jj} x_{j,k} + \gamma s_{j,k} + \sum_{v_i \in \inneighbor{j}} r_{ji} x_{i,k}, \\
s_{j,k+1} &= c_{jj} s_{j,k} + x_{j,k} - x_{j,k+1} + \sum_{v_i \in \inneighbor{j}} c_{ji} s_{i,k},
\end{align}
\end{subequations}
where $\gamma > 0$ is the \emph{surplus gain}\footnote{The choice of $\gamma$ requires global knowledge of the network size.}, used to tune the influence of the surplus variable $s$ on the primal state variable $x$. The weights $r_{ji}$, used to combine received $x_{i,k} $, are given by:
\begin{align}\label{eq:r-weights}
    r_{ji} = \begin{cases}
    \frac{1}{1+\indegree{j}}, & \text{if } v_i \in \inneighbor{j} \lor j = i,\\
    0, & \text{otherwise}.
    \end{cases}
\end{align}

\begin{remark}
\emph{Push weights} $c_{lj} \geq 0$ are assigned by the transmitting agents based on their out-degree, which is either estimated or computed (see~\cite{hadjicostis2015robust,charalambous2015distributed1,makridis2023utilizing}). These weights form a column-stochastic matrix $C \in \mathbb{R}_{\geq0}^{n \times n}$ with $\mathbf{1}^\top C = \mathbf{1}^\top$.
\end{remark}

\begin{remark}
\emph{Pull weights} $r_{ji} \geq 0$ are assigned by the receiving agents using their in-degree, which can be directly counted. These weights form a row-stochastic matrix $R \in \mathbb{R}_{\geq0}^{n \times n}$ with $R\mathbf{1} = \mathbf{1}$.
\end{remark}

The main idea behind \cite{cai2012average} is to keep the sum of agents' variables (state and surplus) time-invariant, and equal to the initial sum of agents' variable, \ie $\mathbf{1}^{\top}({\bm x}_k + {\bm s}_k) = \mathbf{1}^{\top} {\bm x}_0$ for all time $k\in\mathbb{Z}_{\geq0}$, where ${\bm x}_k$ and ${\bm s}_k$ are stacks of the state and surplus variables of all agents into column vectors.

\section{Problem Formulation}\label{sec: problem}

Consider a network of $n$ agents indexed by $\mathcal{V}=\{1,\dots,n\}$,
interconnected over a fixed directed graph (digraph) $\mathcal{G}=(\mathcal{V},\mathcal{E})$. Each agent $v_{j}$ holds an initial scalar value $x_{j,0}\in\mathbb{R}$, and the collective objective of all agents is to compute the average consensus value defined as $x_{\text{ave}}$ in \eqref{eq:ac_problem}, in a distributed manner, and using only local communication and processing. The communication among agents is subject to quantization with a finite and limited number of bits, hence, agents are only allowed to communicate quantized messages. In this paper, we aim to design a distributed coordination algorithm that drives agents' states converge to the exact (asymptotic) average consensus value $x_{\text{ave}}$, despite quantization and information flow imbalance due to the directionality of the links. In addition to exact average consensus, we are also interested in characterizing \emph{$\epsilon$-convergence}, whereby each agent is guaranteed to reach an $\epsilon$-accurate estimate of the global average consensus value in finite time, with a fully distributed and synchronized termination mechanism. 

\begin{remark}
With few quantization bits, some values may become unrepresentable, causing saturation at certain agents. While most existing methods assume enough bits to avoid this issue, the proposed \ouralgorithm{} algorithm can still operate under tight bit-budget. In fact, \ouralgorithm{} tolerates temporary saturation in the exchange of additional communication rounds to gradually refine the quantized information.
\end{remark}

\section{Quantized Average Consensus with Dynamic Framing}
\label{sec:algo}

Assume that each agent $v_j\in\set{V}$ is given an initial quantization step size $\Delta_0$, midpoint $\sigma_0$, and $b$ number of bits for encoding its local state and surplus variables $\check{x}_{j,k}$ and $\check{s}_{j,k}$, so that they are sent to its out-neighbors $v_l\in\outneighbor{j}$, in a quantized form, \ie
\begin{subequations}\label{eq:quantized_values}
    \begin{align}
        \check{x}_{j,k} &= Q\big(x_{j,k},b,\Delta_k,\sigma_k\big),\\
        \check{s}_{j,k} &= Q \big(s_{j,k},b,\Delta_k,0\big),
    \end{align}
\end{subequations}
where $Q(\cdot)$ is enforced to be identical for each agent, by the construction of the algorithm. Here it is important to note that, the quantizer that gives the quantized value $\check{s}_{j,k}$ for the surplus variable $s_{j,k}$, takes the same number of bits and quantization step size as the quantizer for $x_{j,k}$ with the only difference being the midpoint $\sigma$ which is always at $0$. A more detailed discussion regarding the midpoint $\sigma$ is provided in subsequent sections. 

Before introducing the coordinated dynamic framing mechanism for zooming the quantization levels and shifting the quantizers' midpoint, we state the following assumptions that are essential for the description of the proposed algorithm.

\begin{assumption}\label{assum:1}
Let $D$ denote the directed diameter of the network, defined as $D = \max_{v_{i},v_{j} \in\mathcal{V}} \delta(i,j)$, where $\delta(i,j)$ is the length of the shortest directed path from node $v_{i}$ to node $v_{j}$. Each agent $v_j \in \mathcal{V}$ is assumed to know an upper bound $\bar D$ on the diameter, satisfying $\bar D \ge D$.
\end{assumption}

\begin{assumption}\label{assum:2}
    All the agents $v_j\in\set{V}$ set their initial quantization step size to $\Delta_0$ and quantization midpoint to $\sigma_0$.
\end{assumption}

\begin{assumption}\label{assum:3}
    All the agents $v_j\in\set{V}$ are given the same surplus gain $\gamma>0$ and quantization zooming factor $\alpha>0$.
\end{assumption}

Assumption~\ref{assum:1} enables us to design max- and min-consensus protocols that are performed by the agents every $\bar{D}$ time steps. 
These protocols enable agents to take coordinated decisions for zooming out or zooming in the quantization level $\Delta_k$ of their quantizer, every $\bar{D}$ steps, with the aim of achieving faster and more accurate quantized average consensus, compared to fixed quantization strategies. 
Assumption~\ref{assum:1} is standard in the literature (see for example \cite{cady2015finite,charalambous2015distributed2}). 
However, one can employ distributed mechanisms to obtain the diameter of the network either via exact computation \cite{oliva2016distributed} or estimation \cite{garin2012distributed}. 
Assumption~\ref{assum:2} enables nodes to maintain efficient communication during the operation of our algorithm over the bandwidth-limited communication links. Assumption~\ref{assum:3} is necessary for the operation of our algorithm enabling consistent coordination with quantized message exchange among nodes.

\subsection{Coordinated Quantizers}\label{sec:coordinated_quantizers}
In what follows we will introduce (a) a coordinated zooming mechanism for adjusting the quantization step size $\Delta_{k}$ that define the quantization levels to encode the state $x_{j,k}$ and surplus variable $s_{j,k}$; and (b) a coordinated midpoint shifting mechanism for adjusting the quantizers' symmetry point $\sigma_{k}$ so that it tracks the average consensus value with the aim of reaching exact (asymptotic) convergence to $x_{\text{ave}}$ in \eqref{eq:ac_problem}. Recall that all agents have access to the same initial quantization step size $\Delta_0$, midpoint $\sigma_0$, and number of bits for quantization $b$. Executing max/min-consensus procedures which converge to the maximum/minimum value exchanged in the network after at most $\bar{D}$ steps, agents are able to coordinate their zooming and shifting mechanisms for adjusting the quantization step size $\Delta_k$ and midpoint $\sigma_{k}$, respectively. 

Before we formally define the coordination mechanism for quantization zooming, we let $\set{K}_{\bar{D}}$ denote the set of times when agents may adjust their quantization step size, \ie
\begin{align}
    \set{K}_{\bar{D}} := \{ k \in \mathbb{Z}_{\geq 0} \mid k = \bar{D}h, \, h = 1,2,3,\ldots \}.    
\end{align}
Here, it is important to emphasize that the broadcast of the quantized values $\check{x}_{j,k}$ and $\check{s}_{j,k}$ from each node $v_{j}\in\set{V}$ (towards converging to the average consensus value) occurs at each time step $k\in\mathbb{Z}_{\geq0}$. In a similar manner, the exchange of coordination values for the max/min-consensus protocols that will be described in the subsequent section, \ie $w_{j,k}$, $m_{j,k}$, and $M_{j,k}$, are broadcast at each time step $k\in\mathbb{Z}_{\geq0}$, although at each time $k\in \set{K}_{\bar{D}}$ the coordination variables are reset. Fig.~\ref{fig:timeline} depicts an example of the time instances when the quantization step size $\Delta_{k}$ (and similarly the midpoint $\sigma_{k}$) is adjusted.
\begin{figure}[h]
    \centering
    \includegraphics[width=0.5\linewidth]{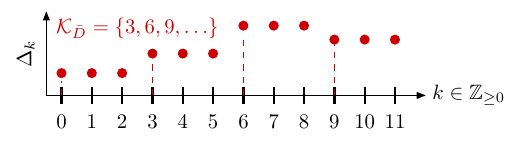}
    \caption{Example of time instances for zooming (resp. midpoint shifting) updates via adjustments of the step size $\Delta_{k}$ (resp. $\sigma_{k}$), for $\bar{D}=3$.}
    \label{fig:timeline}
\end{figure}

\subsection{Coordinated Quantization Zooming}\label{sec:zooming}
Adjusting the quantization step size requires agents to coordinate in order to agree on their quantizers' operation (\eg zoom-in, hold, or zoom-out), and its adjustment on the step-size which determines the quantization precision. All the agents synchronously update the quantization step-size at each $k\in\mathcal{K}_{\bar{D}}$, following a distributed coordination protocol (which we describe later on) that allow agents to agree on the same variable $\zeta_k\in\{-1,0,1\}$. Then, at time $k\in\mathcal{K}_{\bar{D}}$ all agents update their quantization step-size which is identical for all agents, as follows:
\begin{align}\label{eq:delta_update}
    \Delta_{k} =
    \begin{cases}
      (1+\alpha)\,\Delta_{k-1}, & \text{if } \zeta_{k} = 1,\\
      \Delta_{k-1}/(1+\alpha), & \text{if } \zeta_{k} = -1,\\
      \Delta_{k-1}, & \text{otherwise.} 
    \end{cases}
\end{align}
where $\alpha>0$ is the predefined zooming factor. Here, $\zeta_k\in\{-1,0,1\}$ captures the agreed adjustment (zoom-in, hold, zoom-out, respectively), and results from a distributed max-consensus protocol executed during the previous synchronization interval (during the last $\bar{D}$ iterations). To this end, each agent $v_j\in\mathcal V$ maintains a local coordination variable $w_{j,k}\in\{-1,0,1\}$, initialized at $w_{j,0}=0$.
The local variable $w_{j,k}$ classifies the alignment of the value $x_{j,k}$ (which is to be quantized) into three regions, according to the quantization range $q_k:=(2^{b-1}-0.5)\Delta_k$ and the midpoint $\sigma_k$ as shown in Fig.~\ref{fig:regions}. Note that, the midpoint $\sigma_k$ is available by the time that $w_{j,k}$ is determined. More details regarding the update of the midpoint $\sigma_k$ are discussed in Section~\!\ref{sec:midpoint_shifting}.
\begin{figure}[h]
    \centering
    \includegraphics[scale=0.9]{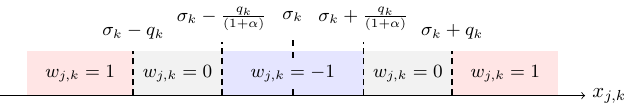}    \caption{Zooming regions. Blue region: central--well within range; Gray: peripheral--near the range boundaries; Red: out-of-range-- exceeding the range.}
    \label{fig:regions}
\end{figure}

During each time step $k\in\mathbb{Z}_{\geq0}\setminus\mathcal{K}_{\bar{D}}$, each agent $v_j\in\set{V}$ broadcasts its local $w$-value to its out-neighbors $v_\ell\in\outneighbor{j}$ (using two bits) and receives the $w$-values from its in-neighbors $v_i\in\inneighbor{j}$ to update its $w_{j,k}$ using a max-consensus protocol, \ie
\begin{align}\label{eq:max_consensus_repeat}
    w_{j,k+1} = \max_{v_i\in\inneighbor{j}\cup\{v_{j}\}}\{\,w_{i,k}\,\},
\end{align}
where each agent also includes its own value to take the max value. As this update is iterated over time, after $\bar{D}$ steps the information has propagated throughout the network and all agents reach agreement on $\max_{v_j\in\set{V}} \{\, w_{j,k} \,\},$
for all $k\in\mathcal{K}_{\bar{D}}=\{\bar{D},2\bar{D},3\bar{D},\ldots\}$. At these synchronization instants $k\in\mathcal{K}_{\bar{D}}$, all agents $v_j\in\set{V}$ set $\zeta_{k}=w_{j,k}$, and subsequently reset their $w$-variables for the next quantizers' coordination interval as
\begin{align}\label{eq:w_reset}
    w_{j,k} &= \begin{cases}
        \hfill 1, & \text{if}~\lvert x_{j,k} - \sigma_{k} \rvert > q_{k},\\
        \hfill -1, & \text{if}~\lvert x_{j,k} - \sigma_{k} \rvert < q_{k}/(1+\alpha),\\        
        \hfill 0, & \text{otherwise}.
    \end{cases}
\end{align}
Here, it is worth mentioning that the reset of $w$-variable is based on local checks that classify $w_{j,k}$ into the three regions according to the state variable $x_{j,k}$ and the quantization range $q_{k}$. This reset initializes a new round of max-consensus process which then runs for the next $\bar{D}$ iterations with a fixed quantizer with step size $\Delta_k$ and midpoint $\sigma_{k}$. Intuitively, $w_{j,k}$ evolves by max-consensus steps between quantization adjustment instants, and is reinitialized at every multiple of $\bar{D}$ based on the new quantizer configuration.

\subsection{Coordinated Quantization Midpoint Shifting}\label{sec:midpoint_shifting}
To achieve arbitrarily small error on the average consensus value, the symmetry midpoint of the uniform quantizers should be shifted as well at every zooming instance. In essence, we would like to design an update rule for the quantizers' midpoint value $\sigma_k$ that asymptotically tracks the average consensus value in \eqref{eq:ac_problem}, in a distributed fashion. This allows agents to further zoom-in their quantizers such that they reach average consensus with arbitrary high precision, without reaching a plateau.

To shift the quantizers' midpoint $\sigma_k$ in a distributed and coordinated manner, each agent $v_j \in \set{V}$ maintains two auxiliary variables $M_{j,k}$ and $\mu_{j,k}$ which are updated using max- and min-consensus, respectively. The max- and min-consensus mechanism is similar to the one introduced in Section~\!\ref{sec:zooming} for quantization zooming. In particular, at each time $k \in \mathbb{Z}_{\ge 0}$, agent $v_j$ broadcasts $M_{j,k}$ and $m_{j,k}$ (with $b$ bits each) to all its out-neighbors 
$v_\ell \in \outneighbor{j}$ and simultaneously receives $M_{i,k}$ and $m_{i,k}$ from all its in-neighbors 
$v_i \in \inneighbor{j}$ to compute the local maximum according to
\begin{subequations}\label{eq:max-min-consensus}
\begin{align}
    M_{j,k+1} &= \max_{v_i \in \inneighbor{j}\cup \{v_j\}}\big\{M_{i,k}\big\},\\
    \mu_{j,k+1} &= \min_{v_i \in \inneighbor{j}\cup \{v_j\}}\big\{\mu_{i,k}\big\},
\end{align}
\end{subequations}
with initial conditions $M_{j,0}=0$ and $m_{j,0}=0$ for all $v_{j}\in\set{V}$. Note that, these variables are reset at each time $k \in \set{K}_{\bar{D}}$, as $M_{j,k}=\check{x}_{j,k}$ and $\mu_{j,k}=\check{x}_{j,k}$. By the end of ${\bar{D}}$ iterations of the aforementioned max- and min-consensus procedures, agents compute the maximum and minimum quantized value in the network, \ie $M_{k}=\max_{v_j \in \set{V}}\{M_{j,k}\}$ and $\mu_{k}=\min_{v_j \in \set{V}}\{\mu_{j,k}\}$, respectively. Based on these values, they shift their quantizers' midpoint according to:
\begin{align}\label{eq:shifting}
    \sigma_k = \frac{1}{2}(M_k + \mu_k).
\end{align}

\begin{remark}
It is worth noting that the quantizer parameters, namely the step size $\Delta_k$ and midpoint $\sigma_k$, 
are updated only at synchronization instants $k \in \mathcal{K}_{\bar{D}}$. Hence, during the intermediate $\bar{D}$ iterations both parameters remain fixed.
\end{remark}

\subsection{Distributed Averaging with Coordinated Quantization}

In what follows, we describe the local computations executed by each agent $v_j \in \mathcal{V}$ at each time $k \in \mathbb{Z}_{\ge 0}$ using \textbf{Algorithm\!~\ref{alg:alg1}}. Recall that, each agent maintains the local variables 
$x_{j,k}$ and $s_{j,k}$ representing, respectively, the state variable and the auxiliary variable, and 
communicates only quantized versions of these quantities to its out-neighbors.

\begin{center}
\begin{minipage}{.5\linewidth}
\begin{algorithm}[H]
\caption{\ouralgorithm{} from agent's $v_j$ perspective.}
\label{alg:alg1}
\small
\KwIn{$x_{j,0} \in \mathbb{R}$, $s_{j,0} = 0$, $\sigma_0=0$, $\Delta_0=1$, $\zeta_0=0$, $w_{j,0}=0$, $M_{j,0}=0$, $m_{j,0}=0$, $\gamma>0$}
\For{ $k = 0,1,2,\ldots$}{
    \textbf{Update:} {$(\Delta_{k},\sigma_{k})$ via Algorithm\!~\ref{alg:alg2}} \textbf{if} $k>0$
    
    \textbf{Compute:} 
    \begin{flalign*}
        &\check{x}_{j,k} = Q(x_{j,k},b,\Delta_k,\sigma_k)&\\
        &\check{s}_{j,k} = Q(s_{j,k},b,\Delta_k,0)&\\
        &e^x_{j,k} = x_{j,k} - \check{x}_{j,k}&\\
        &e^s_{j,k} = s_{j,k} - \check{s}_{j,k}&
    \end{flalign*}
    
    \textbf{Broadcast:}
    $\check{x}_{j,k}$ and $\check{s}_{j,k}$ to all $\ell \in \outneighbor{j}$
    
    \textbf{Receive:}
    $\check{x}_{i,k}$ and $\check{s}_{i,k}$ from all $i \in \inneighbor{j}$

    \textbf{Update:}   
    \begin{flalign*}
    & x_{j,k+1} = \!\sum_{v_i \in \inneighbor{j}} \! r_{ji} \check{x}_{i,k} + \gamma\, s_{j,k} + e^x_{j,k}&\\
    & s_{j,k+1} = \!\sum_{v_i \in \inneighbor{j}} \! c_{ji} \check{y}_{i,k} - x_{j,k+1} + x_{j,k} + e^{s}_{j,k}&
    \end{flalign*}   
}
\KwOut{$x_{j,k+1}$}
\end{algorithm}
\end{minipage}
\end{center}

Each agent initializes its local states as $x_{j,0} \in \mathbb{R}$ and $s_{j,0}=0$, 
and sets the quantizer parameters $\sigma_0=0$, $\Delta_0=1$, $\alpha>0$, and $\zeta_0=0$. 
The initial quantizer $Q(\cdot)$ is thus centered at $\sigma_0$ with step size $\Delta_0$. At each iteration $k\in\mathbb{Z}_{\geq0}$ each agent encodes its current values using the local coordinated quantizer $Q(\cdot)$, to obtain $\check{x}_{j,k}$ and $\check{s}_{j,k}$ as in \eqref{eq:quantized_values}, and computes the corresponding quantization errors
\begin{subequations}\label{eq:quantization_errors}
\begin{align}
e^x_{j,k} &= x_{j,k} - \check{x}_{j,k},\\
e^s_{j,k} &= s_{j,k} - \check{s}_{j,k}.    
\end{align}
\end{subequations}
Here, $Q(\cdot)$ corresponds to the last updated configuration of the coordinated quantizer with $b$ number of bits, step size $\Delta_k$, and midpoint $\sigma_k$ for $x$-variable (and $0$ for $s$-variable). Hence, each agent $v_j\in\set{V}$ encodes the values that are to be broadcast, based on the last updated quantizer with parameters $(\Delta_{k},\sigma_{k})$, \ie 
\begin{align}
    Q(u_k,b,\Delta_k,\sigma_k) &= \begin{cases}
        \sigma_{k} + \Delta_{k}(2^{b-1} -1), & \text{if}~u_k > \overline{u}_{k},\\
        \sigma_{k} - \Delta_{k}(2^{b-1} -1), & \text{if}~u_k\leq\underline{u}_{k},\\
        \sigma_{k} + \Delta_{k} \bigl \lfloor \frac{u_{k}-\sigma_{k}}{\Delta_{k}} +\frac{1}{2}\bigr \rfloor, & \text{otherwise},
    \end{cases}
\end{align} 
where $u_k$ is the input value at time $k$ to be encoded (\ie $x_{j,k}$ and $s_{j,k}$), and $\overline{u}_k$ and $\underline{u}_k$ are the upper and lower quantization limits, respectively, which are given by
\begin{subequations}
\begin{align}
    \overline{u}_{k} = \sigma_{k} + \Delta_{k}(2^{b-1}-0.5),\\
    \underline{u}_{k} = \sigma_{k} - \Delta_{k}(2^{b-1}-0.5).  
\end{align}
\end{subequations}
Subsequently, each agent $v_j\in\set{V}$ broadcasts $\check{x}_{j,k}$ and $\check{s}_{j,k}$ to all its out-neighbors 
$v_{\ell} \in \outneighbor{j}$, and receives the corresponding values 
$\check{x}_{i,k}$, $\check{s}_{i,k}$ from all its in-neighbors 
$v_i \in \inneighbor{j}$. Upon the reception and decoding of this information, each agent $v_j\in\set{V}$ updates its states as follows:
\begin{subequations}
\begin{align}
x_{j,k+1} &=  \sum_{v_i \in \inneighbor{j}} r_{ji}\check{x}_{i,k} + x_{j,k} + \gamma s_{j,k},\\
s_{j,k+1} &= \sum_{v_i \in \inneighbor{j}} c_{ji}\check{s}_{i,k} + s_{j,k} - x_{j,k+1} + x_{j,k},
\end{align}
\end{subequations}
where $c_{ji}$ and $r_{ji}$ are given in \eqref{eq:c-weights} and \eqref{eq:r-weights}, respectively. Note that, the quantized value $\check{s}_{i,k}$ is sent over the communication channel at each iteration $k\in\mathbb{Z}_{\geq0}$, while its corresponding weight $c_{ji}$ can be reconstructed exactly by agent $v_j$, by simply setting $c_{ji}=1/(d_i^{\texttt{out}}+1)$, once the out-degree of agent $v_i$, $d_i^{\texttt{out}}\in\mathbb{N}$, is available at node $v_{j}$ at the beginning of the consensus iterations.

The quantizer parameters $(\Delta_k,\sigma_k)$ are updated by \textbf{Algorithm\!~\ref{alg:alg2}}, which performs the distributed coordination and adjustment of the quantization step size based on the consensus variable $\zeta_k$. 
This ensures that all agents synchronously zoom in, zoom out, or hold the quantization range according to network-wide conditions.

\begin{center}
\begin{minipage}{.5\linewidth}
\begin{algorithm}[H]
\caption{Distributed quantizer with dynamic framing at time $k$ from agent's $v_j$ perspective.}
\label{alg:alg2}
\small
\KwRequire{Zooming factor $\alpha>0$}
\KwIn{$x_{j,k}$, $s_{j,k}$, $w_{j,k}$, $M_{j,k}$, $m_{j,k}$, $\sigma_k$, $\Delta_k$}

\textbf{Broadcast:} $w_{j,k}$, $M_{j,k}$, and $m_{j,k}$ to all $v_{\ell}\!\in\!\outneighbor{j}$

\textbf{Receive:} $w_{i,k}$, $M_{i,k}$, and $m_{i,k}$ from all $v_i\!\in\!\inneighbor{j}$

\If{$k \in \mathcal{K}_{\bar{D}}=\{\bar{D},2\bar{D},3\bar{D},\ldots\}$}{

    \textbf{Stop iterations if:} $M_{j,k}-m_{j,k} \leq \epsilon - \Delta_{k}$
    
    \textbf{Set adjustment decision:} $\zeta_{k} = w_{j,k}$

    \textbf{Update step-size:}
    \begin{flalign*}
        &\Delta_{k} =
        \begin{cases}
            (1+\alpha)\Delta_{k-1}, & \text{if } \zeta_{k} = 1,\\
            \Delta_{k-1}/(1+\alpha), & \text{if } \zeta_{k} = -1,\\
            \Delta_{k-1}, & \text{otherwise.}
        \end{cases}&
    \end{flalign*}

    \textbf{Update agreed midpoint:} $\sigma_{k} = m_{j,k}$

    \textbf{Compute dynamic range:} $q_{k} = \big(2^{b-1} - \tfrac{1}{2}\big)\Delta_{k}$

    \textbf{Reset coordination variables:}
    \begin{flalign*}
        &w_{j,k} =
        \begin{cases}
            1, & \text{if } |x_{j,k} - \sigma_{k}| > q_{k},\\
            -1, & \text{if } |x_{j,k} - \sigma_{k}| < q_{k}/(1+\alpha),\\
            0, & \text{otherwise,}
        \end{cases}&\\
        &M_{j,k} = \check{x}_{j,k},&\\
        &m_{j,k} = \check{x}_{j,k}.&
    \end{flalign*}}
\Else{
\textbf{Hold: $(\Delta_{k}\leftarrow\Delta_{k-1},\sigma_{k}\leftarrow\sigma_{k-1})$
}}

\textbf{Update:}
\begin{flalign*}
    &w_{j,k+1} = \max_{v_i\in\inneighbor{j}\cup\{v_{j}\}} \{ w_{i,k} \},&\\
    &M_{j,k+1} = \max_{v_i\in\inneighbor{j}\cup\{v_{j}\}} \{ M_{i,k} \},&\\
    &m_{j,k+1} = \max_{v_i\in\inneighbor{j}\cup\{v_{j}\}} \{ m_{i,k} \}.&
\end{flalign*}

\KwOut{$(\Delta_{k}, \sigma_{k})$}
\end{algorithm}
\end{minipage}
\end{center}

\begin{remark}
    Our proposed algorithm (\textbf{Algorithm\!~\ref{alg:alg1}}) adopts the surplus structure of \cite{cai2012average}, and is a distributed average consensus variant of the Push-Pull Gradient Tracking method introduced by Song \emph{et al.}~\cite{song2022compressed}, which addresses distributed optimization with compression over general directed networks. Although the algorithm in \cite{song2022compressed} can solve the average consensus problem as a special case, it is specifically designed for optimizing global objective functions with local information at each node. As a result, it introduces unnecessary computational and communication overhead for consensus tasks, such as maintaining additional variables, performing gradient evaluations, and communicating them over the network. In contrast, our algorithm directly targets distributed average consensus using a dynamic quantization scheme that achieves arbitrarily small error with a fixed number of bits, and for any initial $x_{j,0}\in\mathbb{R}$ for each $v_{j}\in\set{V}$.
\end{remark}

\subsection{Distributed Coordinated Termination}
Now, let us introduce a distributed scheme to enable agents in a network to simultaneously stop transmitting and updating their values as soon as they have reached a predefined averaging accuracy. Finite-time convergence to an $\epsilon$-neighborhood of the average consensus value is essential as practical multi-agent systems should know when to stop exchanging messages to save bandwidth, energy, and computation. Moreover, it provides a guaranteed accuracy needed by subsequent tasks that rely on the consensus value \cite{makridis2024fully,fu2025cutting}. A distributed termination mechanism provides a principled way for all agents to trigger, using only local checks, that the network has reached the desired $\epsilon$-consensus around the global average, and to terminate communication synchronously. Such termination mechanisms are essential in quantized average consensus as (i) exact convergence is unattainable, (ii) asymptotic convergence requires unlimited communication and energy, and (iii) no central coordinator exists to declare termination.

In this setting, we require all agents $v_j\in\mathcal{V}$ to stop transmitting and updating their values as soon as
\begin{align}
	|x_{j,k} - x_{\text{ave}}| \leq \epsilon,
\end{align}
where $\epsilon$ is a known error bound specified \emph{a~priori} by all agents  $v_j\in\mathcal{V}$. A distributed mechanism that achieves stopping in directed graphs without a central coordinator can be found in \cite{cady2015finite}. However, such a mechanism would require agents to transmit extra real-valued information to their out-neighbors. Instead, here, we exploit the coordinated quantizer to allow agents synchronously stop their average consensus procedure. More specifically, at each time $k\in\mathcal{K}_D$ and before resetting the max- and min-consensus variables $M_{j,k}$ and $\mu_{j,k}$, each agent $v_j\in\mathcal{V}$ executes the following check 
\begin{align}\label{eq:stop_condition}
	M_{j,k} - \mu_{j,k} \leq \epsilon - \Delta_k,
\end{align}
which can be seen as a local check since $\Delta_k$ is coordinated, \ie same for all agents. This local check establishes that the agents' values (max and min values) are already close enough that, even after accounting for the current quantization error, the consensus error is already within the desired accuracy $\epsilon$. Recall that, $M_{j,k}$ and $\mu_{j,k}$ are guaranteed to be the same among the agents as a result of the coordinated quantization midpoint shifting mechanism. Thus, the condition in \eqref{eq:stop_condition} holds true synchronously for all agents, and they can stop transmitting and updating their values. In what follows, we provide a formal statement that describes the local stopping criterion.

\section{Convergence Analysis}
To analyze the convergence properties of our proposed algorithm, we first introduce its compact vector-matrix form which can be written as:
\begin{subequations}\label{eq:qac_az}
\begin{align}
{\bm x}_{k+1} &= {\bm x}_k + \gamma{\bm s}_k + (R-I) \check{\bm {x}}_k,\label{eq:qac_az_1}\\
{\bm s}_{k+1} &= {\bm x}_k - {\bm x}_{k+1} + {\bm s}_k + (C-I) \check{\bm s}_k\label{eq:qac_az_2}.
\end{align}
\end{subequations}

A necessary requirement for any average consensus algorithm is the preservation of the global state sum (often called the \emph{total mass}), ensuring that the average value is conserved over time. The following lemma shows that our quantized update dynamics maintain this property despite the presence of instantaneous quantization errors.

\begin{lemma}\label{lemma:1}
    The total mass in the network (hence the state average) is preserved using the iterations in \eqref{eq:qac_az} for all $k\in\mathbb{Z}_{\geq0}$, \ie $\mathbf{1}^{\top}\big( {\bm x}_{k} + {\bm s}_{k} \big) = \mathbf{1}^{\top} {\bm x}_{0}$. 
\end{lemma}
\begin{proof}
    Consider the summation of all agents' $x$ and $s$ values at time step $k+1$, \ie $\mathbf{1}^{\top}\big( {\bm x}_{k+1} + {\bm s}_{k+1} \big)$. Then, plugging  \eqref{eq:qac_az_1} into \eqref{eq:qac_az_2} and substituting into the evolution at time step $k+1$, we get:
    \begin{align}
        \mathbf{1}^{\top}\big( {\bm x}_{k+1} + {\bm s}_{k+1} \big) \nonumber = & \mathbf{1}^{\top} {\bm x}_{k} + \mathbf{1}^{\top} {\bm s}_{k} + \mathbf{1}^{\top} C \check{{\bm s}}_{k} - \mathbf{1}^{\top} I \check{{\bm s}}_{k} \nonumber\\ = &
        \mathbf{1}^{\top} \big( {\bm x}_{k} + {\bm s}_{k} \big),
    \end{align}
    where the last equality comes from the column-stochasticity of matrix $C$. This proves the statement of the lemma.
\end{proof}

Having established the mass-preservation property of the proposed quantized dynamics, we next present the convergence guarantees of \textbf{Algorithm\!~\ref{alg:alg1}}. Theorem~\ref{theorem:1} shows that the algorithm achieves average consensus asymptotically, as the quantization errors vanish when a sufficiently small zooming factor $\alpha>0$ is selected in the dynamic quantization framing scheme. Theorem~\ref{theorem:2} establishes $\epsilon$-convergence under the proposed finite-time termination rule, which allows all agents to synchronously detect when their states have entered an $\epsilon$-neighborhood of the true average. Before introducing the main theorems of this paper, we further define the following notation which is essential for the corresponding convergence proofs.

Substituting \eqref{eq:qac_az_1} into \eqref{eq:qac_az_2}, and augmenting the state by letting ${\bm z}_k = \begin{bmatrix}{\bm x}_k^{\top}, {\bm s}_k^{\top}\end{bmatrix}^{\top}\in\mathbb{R}^{2n}$ and $\check{\bm z}_k = \begin{bmatrix}\check{\bm x}_k^{\top}, \check{\bm s}_k^{\top}\end{bmatrix}^{\top}\in\mathbb{R}^{2n}$ the iterations of the proposed \ouralgorithm{} can be rewritten in the following form: 
\begin{align}\label{eq:aug_form}
{\bm z}_{k+1}  = 
\begin{bmatrix}
    I_n & \gamma I_n\\
    0_{n\times n} & (1-\gamma)I_n
\end{bmatrix} {\bm z}_{k} + \begin{bmatrix}
    R-I_n & 0_{n\times n}\\
    I_n-R & C-I_n
\end{bmatrix} \check{\bm z}_{k}.
\end{align}

Now, considering the quantization error we can define $\check{\bm z}_{k} \triangleq {\bm z}_k + {\bm e}_k$, to further obtain
\begin{align}\label{eq:aug_form_error}
    {\bm z}_{k+1} = (\Gamma + \Pi) {\bm z}_{k} + (\Pi - I_{2n}){\bm e}_k,
\end{align}
where $\Gamma,\Pi\in\mathbb{R}^{2n\times 2n}$, with
\begin{align}
    \Gamma = \begin{bmatrix}  0_{n\times n} & \gamma I_n \\  0_{n\times n} & - \gamma I_n \end{bmatrix}, \quad \text{and} \quad \Pi = \begin{bmatrix} R & 0_{n\times n} \\ I_n-R & C \end{bmatrix}.
\end{align}

\begin{theorem}\label{theorem:1}
    The \ouralgorithm{} algorithm in \eqref{eq:qac_az} achieves exact average consensus, \ie ${\bm x}_{k} \rightarrow x_{\text{ave}}\mathbf{1}$ and ${\bm s}_{k} \rightarrow \mathbf{0}$ as $k\rightarrow\infty$, with a properly chosen surplus gain $\gamma>0$ and sufficiently small zooming factor $\alpha>0$.
\end{theorem}

\begin{proof}\label{sec:appendix} 
The proof can be found in the appendix.
\end{proof}

\begin{theorem}[]\label{theorem:2}
	The \ouralgorithm{} algorithm in \eqref{eq:qac_az} achieves $\epsilon$-consensus to the average in some finite time $k^*\in\mathcal{K}_D$, \ie 
	\begin{align}
		\big| x_{j,k^*} - x_{\textnormal{ave}} \big| \leq \epsilon,\quad \forall v_j\in\mathcal{V},
	\end{align}
	 and synchronously stop their iterations if at iteration $k^*$ every agent's state satisfies
	\begin{align}
		M_{k^*} - \mu_{k^*} \leq \epsilon - \Delta_{k^*}.
	\end{align}	
\end{theorem}

\begin{proof}
	First, let us introduce the following notation. Define $x^{\max}_k := \max_j x_{j,k}$ and $x^{\min}_k := \min_j x_{j,k}$. From the construction of the quantizer and the coordinated quantization midpoint shifting mechanism, we know that 
	$$
		\big| (x^{\max}_k - x^{\min}_k) - (M_k - \mu_k) \big| \leq \Delta_k,
	$$
	since $\check{x}_{i,k}-x_{i,k} \leq \Delta_k/2$. Thus, if $M_k - \mu_k \leq \epsilon - \Delta_k$, then 
	$$
		x^{\max}_k - x^{\min}_k \leq (M_k - \mu_k) + \Delta_k \leq \epsilon.
	$$
	From Lemma~\ref{lemma:1}, we know that the total mass in the network is preserved, so is the average. Hence, if $x^{\min}_k \leq x_{j,k} \leq x^{\max}_k$ for all $v_j\in\mathcal{V}$ and $x^{\min}_k \leq x_{\text{ave}} \leq x^{\max}_k$, then
	$$
	|x_{j,k} - x_{\text{ave}}| \leq x^{\max}_k - x^{\min}_k \leq \epsilon,
	$$
	for all $v_j\in\mathcal{V}$, which states that every agent's value is within $\epsilon$ of the average consensus value. Thus, all agents stop iterating and freeze their states at $k^*$, since $\Delta_k$ is the same for all agents, and thus the network achieves $\epsilon$-consensus to the average. This completes the proof.
\end{proof}

\section{Simulation Results}
\label{sec:simulations}

In this section, we illustrate the behavior of the proposed distributed algorithm and we showcase its performance. Throughout the simulations, we consider a fixed directional network of $n=5$ agents, which we model as a directed graph $\set{G}=(\set{V},\set{E})$, of diameter $D=4$. Under this setup, each agent $v_j\in\set{V}$ knows an upper bound on the diameter of the network, ${\bar{D}}=4$, and initializes its local variables with $\Delta_0=0$, $\sigma_0=0$, $s_{j,k}=0$, $\gamma=0.2$, and $\alpha\in\{0.2,4\}$. Note that, we omit the agents' index $j$ for the variables that are always the same between agents due to either the coordination mechanisms we described in the previous sections, \ie $\Delta_k$ and $\sigma_k,$ and based on the assumption that agents have global information of the parameters $\alpha, \gamma,$ and ${\bar{D}}$, apriori. The initial state variable of agent $v_j$, $x_{j,0}$, is assumed to be randomly initialized from a uniform distribution in the interval $[0,1000]$. In general, to avoid saturation and ensure full coverage of the initial values that $x_{j,0}$ can take between the specified interval, then, agents need at least $10$ bits, \ie $2^{10}>1000$. For a lower number of bits it is not guaranteed that values will be representable, and hence, some agents may experience saturation on their quantized values.  Although the methods in the literature proposed solutions where the number of bits for quantization is enough to represent the information that is to be exchanged, the \ouralgorithm{} algorithm of this work can handle cases where the number of quantization bits are not enough to fully represent a real value, in the expense of more communication rounds. In what follows, we illustrate the evolution of the values of agents in the network, and we show the convergence speed of the agents for different settings. 

\subsection{Asymptotic case}
In Fig.~\ref{fig:alpha_small} and Fig.~\ref{fig:alpha_large} we present the evolution of the values of the agents executing the proposed algorithm, over the network and initial values we described before, and we examine the average consensus error captured by $\lVert {\bm x}_k - \mathbf{1}x_{\text{ave}}\rVert_2$, for $\alpha=0.2$ and $\alpha=4$, respectively. As shown in these figures, agents converge faster to the average consensus value when $\alpha=5$ rather than $\alpha=1.2$. At the first rounds of communication for $b\in\{3,8\}$, agents increase their quantization step sizes $\Delta_k$ every ${\bar{D}}=4$ iterations, since the values of their state variables $x_{j,k}$ cannot be fully represented by $b\in\{3,8\}$ bits. Hence, they increase their step size $\Delta_k$, and shift their midpoint $\sigma_k$ in a distributed and coordinated manner, up to the point where the values they want to send to their out-neighbors are within their quantization range. Clearly, for $b=24$, agents start by zooming-in their quantizers (by decreasing the quantization step size) since with $24$ bits, all their values can be represented adequately. Notice that, for a relatively lower number of bits and higher value of $\alpha$, the quantizers' step size fluctuates since the zooming factor $\alpha$ is too big for the given number of bits.  

\begin{figure}
    \centering
    \includegraphics[width=0.5\linewidth]{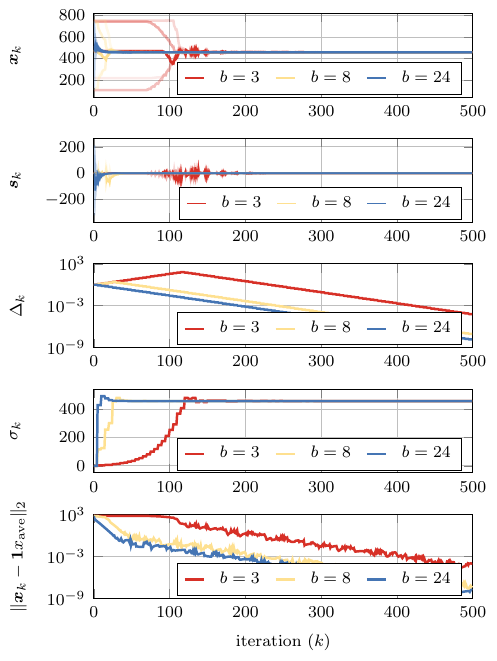}
    \vspace{-10pt}
    \caption{Evolution of the \ouralgorithm{} algorithm over a directed network of $n=5$ agents with $D=4$ and $\alpha=0.2$.}
    \label{fig:alpha_small}
\end{figure}

\begin{figure}
    \centering
    \includegraphics[width=0.5\linewidth]{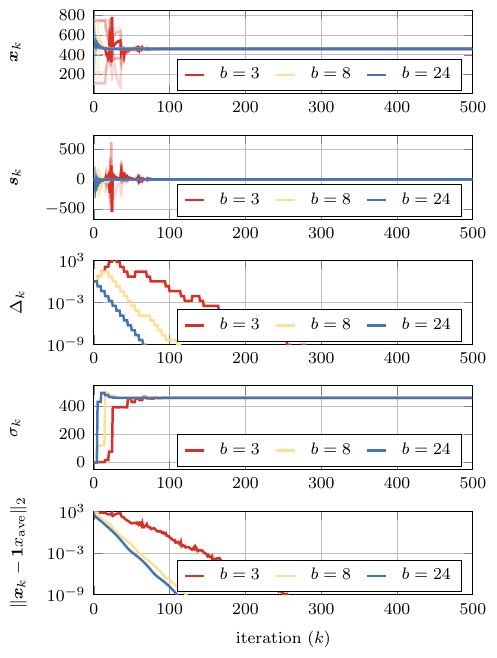}
    \vspace{-10pt}
    \caption{Evolution of the \ouralgorithm{} algorithm over a directed network of $n=5$ agents with $D=4$ and $
    \alpha=4$.}
    \label{fig:alpha_large}
\end{figure}

We now analyze the performance of the \ouralgorithm{}  algorithm in terms of average convergence time and the number of communication/quantization bits used. Specifically, we conduct $50$ Monte Carlo simulations on two different sizes of networks $n\in\{5,20\}$ with varying connectivity. The differences in connectivity correspond to variations in the network's diameter $D$, which is assumed to be known by all agents. To capture the convergence speed of the \ouralgorithm{} algorithm, we consider such networks with average diameter $\hat{D}=3.6$ (for $n=5$) and $\hat{D}=7.66$ (for $n=20$), as shown in the left and right plot of Fig.~\ref{fig:bits_convergence_net1}, respectively. In these simulations, we try different quantizer zooming factors $\alpha\in\{0.2,0.3,0.4,0.6,1.0,4.0\}$ and quantization bits $b\in\{2,4,6,8,10,12,14,16\}$ for all agents, and record the average convergence time, \ie $\min\big\{k \mid \lvert x_{i,k}-x_{j,k}\rvert\leq 10^{-8}, \forall v_j, v_i \in \mathcal{V}\big\}$.

The average convergence rate of the networks with average diameter $\hat{D}=3.6$ is significantly faster than that of networks with $\hat{D}=7.66$, primarily due to the more frequent adjustments in quantization step size and midpoint shifting. However, notice that, selecting a higher quantization zooming factor, such as $\alpha=4$, leads to the saturation of the \ouralgorithm{} algorithm, preventing further improvement in quantizer precision. This saturation arises from excessive fluctuations in the quantization step size, as shown in Fig.~\ref{fig:alpha_small} and Fig.~\ref{fig:alpha_large}, which inhibit the zooming mechanism from effectively reducing it. Thus, in these cases the algorithm cannot reach $\lvert x_{i,k}-x_{j,k}\rvert\leq 10^{-8}, \forall v_j, v_i \in \mathcal{V}$, and hence the average convergence time is not shown in Fig.~\ref{fig:bits_convergence_net1} for some instances of $b=2$ number of bits. Observe that, by increasing the number of bits for communication, the convergence time of the \ouralgorithm{} algorithm decreases significantly, especially with smaller zooming factor $\alpha$. 

\begin{figure}[t]
    \centering
    \includegraphics[width=0.6\linewidth]{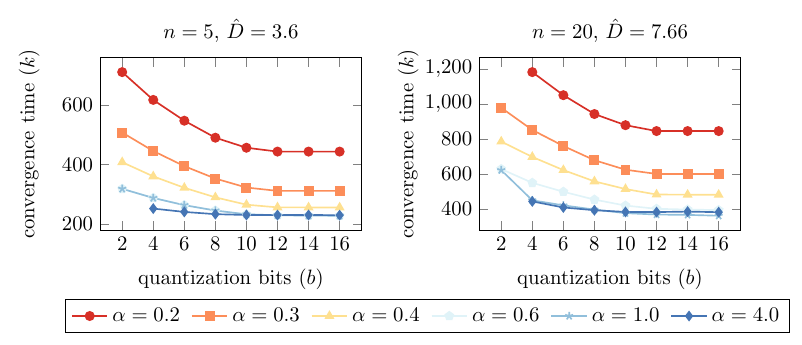}
    \vspace{0pt}
    \caption{Average convergence time of the \ouralgorithm{} algorithm over 50 Monte Carlo simulations, \ie $\min\{k \mid \lvert x_{i,k}-x_{j,k}\rvert\leq 10^{-8}, \forall v_j, v_i \in \mathcal{V}\}$ with $\gamma=0.1$. Left: $5$-agent directed networks with average diameter $\hat{D}=3.6$. Right: $20$-agent directed networks with average diameter $\hat{D}=7.66$.}
    \label{fig:bits_convergence_net1}
\end{figure}

\subsection{Finite-time $\epsilon$-consensus case}
We now present numerical results that validate the performance of the \ouralgorithm{} algorithm when applied to the $\epsilon$-convergence average consensus problem with distributed stopping. In particular, we examine how the quantization parameters, \ie the number of bits $b$ and the zooming factor $\alpha$, influence both the convergence time and the overall communication load, providing us more insights on the interplay between quantizer adaptation, midpoint synchronization, and the distributed stopping mechanism over directed networks.

Fig.~\ref{fig:monte_carlo1} illustrates the convergence behavior of the \ouralgorithm{} algorithm over $50$ Monte Carlo simulations performed over a directed network of $n=5$ agents with average diameter $\hat{D}=3.6$. Here, we investigate how both the number of quantization bits $b\in\{2,4,6,8,10,12,14,16\}$ and the zooming factor $\alpha\in\{0.2,0.3,0.4,0.6,1,4\}$ affect the convergence time under two levels on the average consensus convergence accuracy, \ie $\epsilon=10^{-2}$ and $\epsilon=10^{-6}$, shown in the left and right subplots of Fig.~\ref{fig:bits_convergence_net1}, respectively. In both accuracy levels, the convergence time decreases noticeably as the number of quantization bits increases, with the improvement being particularly significant for smaller values of $\alpha$. In these cases, the step-size adaptation remains stable, allowing the quantizer to zoom in without fluctuations, while for larger zooming factors, \eg $\alpha=4$, the aforementioned oscillatory behavior in the quantization step-size is induced, often triggering repeated zoom-outs and intermittent saturation. These fluctuations hinder the shrinking of the quantization step and slow down the dynamics, which in several instances prevents the algorithm from reaching the required accuracy (\eg for $b=2$ and $\alpha=4$ the corresponding convergence times are absent in Fig.~\ref{fig:monte_carlo1} due to instability). 

Fig.~\ref{fig:monte_carlo2} depicts the (average over $50$ Monte Carlo simulations) total number of transmitted bits required to reach the aforementioned accuracy levels, \ie $\epsilon=10^{-2}$ and $\epsilon=10^{-6}$. Here, the communication cost grows with the number of quantization bits since each transmitted message carries more information. For small zooming factors the algorithm converges substantially faster, resulting in markedly fewer total transmitted bits compared to larger $\alpha$ values. However, similarly to the convergence time behavior, aggressive zooming (\eg $\alpha=4$) leads to oscillatory behavior in the quantization step-size, causing a significant increase in communication load, especially under stricter tolerance $\epsilon=10^{-6}$ that would require more iterations for convergence.


\begin{figure}
	\centering
	\includegraphics[width=0.6\linewidth]{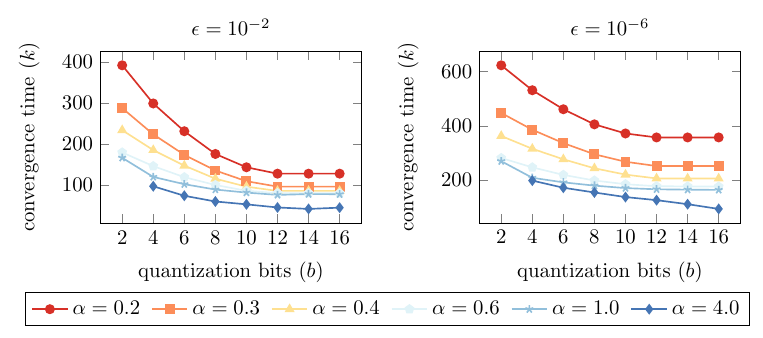}
	\vspace{0pt}
	\caption{Convergence time of the \ouralgorithm{} algorithm with distributed stopping at $\epsilon$-consensus over 50 Monte Carlo simulations in a $5$-agent directed network with average diameter $\hat{D}=3.6$ and $\gamma=0.1$. Left: $\epsilon=10^{-2}$. Right: $\epsilon=10^{-6}$.}
	\label{fig:monte_carlo1}
\end{figure}

\begin{figure}
	\centering
	\includegraphics[width=0.6\linewidth]{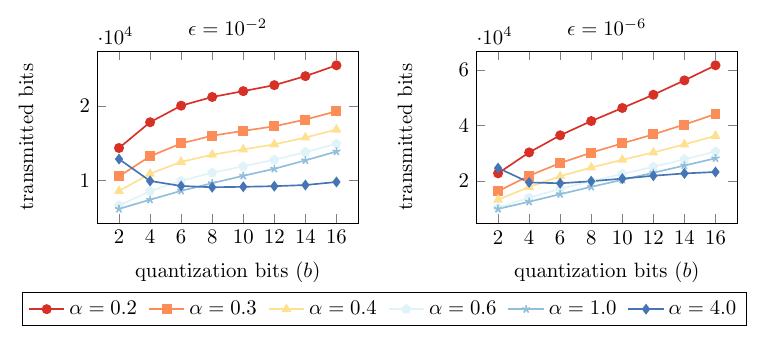}
	\vspace{0pt}
	\caption{Number of transmitted bits using the \ouralgorithm{} algorithm with distributed stopping at $\epsilon$-consensus over 50 Monte Carlo simulations in a $5$-agent directed network with average diameter $\hat{D}=3.6$ and $\gamma=0.1$. Left: $\epsilon=10^{-2}$. Right: $\epsilon=10^{-6}$.}
	\label{fig:monte_carlo2}
\end{figure}

\section{Conclusions and Future Directions}
In this paper, we introduced a deterministic distributed algorithm, referred to as {\ouralgorithm}, for achieving exact average consensus over possibly unbalanced directed graphs using only finite-bit communication. Unlike existing quantized consensus algorithms for unbalanced digraphs, which rely on probabilistic convergence arguments and typically incur non-vanishing quantization errors, our method guarantees deterministic convergence to the exact average while operating with a a fixed and \emph{a~priori} given number of quantization bits. The algorithm integrates Push–Pull (surplus) consensus dynamics with a dynamic quantization framing mechanism that combines zooming and midpoint shifting, enabling agents to preserve the true global average while progressively refining quantization precision so that the consensus error vanishes asymptotically. Beyond asymptotic correctness on any strongly connected digraph, we established a provably sound finite-time termination rule that allows all agents to detect when their states have entered an $\epsilon$-neighborhood of the true average and to stop iterating simultaneously using only local information, ensuring synchronous termination under tight communication constraints. Numerical simulations support the theoretical findings and demonstrate that {\ouralgorithm} achieves reliable and exact performance even under very tight communication bit budgets, making it a promising approach for resource-constrained multi-agent systems operating over directed networks.


Despite the promising results, several avenues remain open for future research. A promising direction is to extend the proposed framework to even more realistic communication settings involving packet drops and time-varying delays, where mass preservation and quantization rules becomes significantly more challenging. Another important direction is to remove the reliance of the agents on the knowledge of the network diameter, enabling fully distributed quantized average consensus without global parameters. Finally, developing mechanisms for asynchronous operation of the algorithm without requiring synchronized updates, would broaden its applicability to practical large-scale and heterogeneous networks.

\appendix\label{sec:appendix}
\subsection*{Proof of Theorem~\ref{theorem:1}:}

In what follows, we analyze the properties of the iteration in \eqref{eq:aug_form_error} based on the quantization error vector ${\bm e}_k$. 

First, consider the case for which ${\bm e}_k = \mathbf{0}~\forall k\in\mathbb{Z}_{\geq 0}$, i.e.,
\begin{align}\label{eq:aug_form_error_app1}
    {\bm z}_{k+1} = (\Gamma + \Pi) {\bm z}_{k},
\end{align}
Since the quantization error is always $0$, the iteration in \eqref{eq:aug_form_error} reduces to the algorithm in \cite{cai2012average} and converges to ${\bm z}^\star = [\mathbf{1}^{\top}x_{\text{ave}}, \mathbf{0}_n^{\top}]^{\top}$ as $k\rightarrow\infty$ \cite[Proposition~6]{cai2012average}.  
    
Next, we consider the case for which ${\bm e}_k = {\bm e} \neq \mathbf{0}~\forall k\in\mathbb{Z}_{\geq 0}$, i.e., 
\begin{align}\label{eq:aug_form_error_app2_1} 
    {\bm z}_{k+1} &= (\Gamma + \Pi) {\bm z}_{k} + (\Pi -I_{2n}){\bm e}. 
\end{align}
For this case, we start by analyzing the matrix $(\Gamma+\Pi)$. 
Defining ${\bm y}_k=[{\bm z}_{k}^\top,{\bm w}_k^\top]^\top$ with ${\bm w}_0={\bm e}$, the system can equivalently be expressed as
\begin{align}
    {\bm y}_{k+1} &= \underbrace{\begin{bmatrix}
        \Gamma+\Pi&\Pi-I_{2n}\\
        0_{2n\times 2n}&I_{2n}
    \end{bmatrix}}_{\Omega} {\bm y}_{k}. \label{eq:augmented_state_dynamics } 
\end{align}
We observe that, being upper block-triangular with the lower diagonal block that corresponds to the $2n\times 2n$ identity matrix, $\Omega\in\mathbb{R}^{2n\times 2n}$ has at least $2n$ eigenvalues equal to one.
Moreover, $\Gamma+\Pi$ has one simple eigenvalue equal to one, while all other eigenvalues are of magnitude less than one \cite{cai2012average}. Therefore, $\Omega$ has $2n+1$ eigenvalues equal to one.
At this point, in order to prove convergence, we claim that the geometric multiplicity of the eigenvalue in one is $2n+1$. To this end, let us characterize the structure of the $2n+1$ linearly independent eigenvectors of $\Omega$ associated to the eigenvalue in one. Specifically, let us consider vectors in the form 
$$
{\bm w}_i=\begin{bmatrix}
    {\bm \eta}_i^{\top},&{\bm 0}_n^\top,&-{\bm \eta}_i^{\top},&{\bm 0}_n^\top
\end{bmatrix}^\top, \quad i\in\{1,\ldots,n\},
$$
where ${\bm \eta}_i$ is the $i$-th vector of any basis of $\mathbb{R}^n$, hence the matrix $E\in\mathbb{R}^{n\times n}$ that has the vectors ${\bm \eta}_i$ as its columns is nonsingular.
We have that
$$
(\Gamma+\Pi)\begin{bmatrix}{\bm \eta}_i\\{\bm 0}_n \end{bmatrix}-(\Pi-I_{2n})\begin{bmatrix}{\bm \eta}_i\\{\bm 0}_n\end{bmatrix}=
(\Gamma+I)\begin{bmatrix}{\bm \eta}_i\\{\bm 0}_n \end{bmatrix}=\begin{bmatrix}{\bm \eta}_i\\{\bm 0}_n \end{bmatrix},
$$
therefore $\Omega {\bm w}_i= {\bm w}_i$ are all eigenvectors of $\Omega$ associated to the eigenvalue in one.

Let us now consider vectors in the form 
$$
{\bm w}_{n+i}=\begin{bmatrix}
    {\bm v}_I^{*\top}-{\bm \eta}_i^{\top},&{\bm 0}_n^\top,&{\bm \eta}_i^{\top},&{\bm v}_{II}^{*\top}
\end{bmatrix}^\top, \quad i\in\{1,\ldots,n\},
$$
where the vectors ${\bm \eta}_i$ are the same as above and
$$
{\bm v}^*=\begin{bmatrix}
    {\bm v}_I^{*\top},&{\bm v}_{II}^{*\top}
\end{bmatrix}^\top,
$$
with ${\bm v}_I^{*}$ being the dominant eigenvector of $R$ while ${\bm v}_{II}^{*}$ is the dominant eigenvector of $C$.
Notice that
$$
\Pi {\bm v}^*= \begin{bmatrix}
    R {\bm v}_I^{*}\\ (I_{n}-R){\bm v}_I^{*}+C{\bm v}_{II}^{*}
\end{bmatrix}=\begin{bmatrix}
    {\bm v}_I^{*}\\ {\bm v}_{II}^{*}
\end{bmatrix},
$$
hence ${\bm v}^*$ is an eigenvector of $\Pi$ associated to the eigenvalue one. At this point, we observe that
$$
\begin{aligned}
 (\Gamma+\Pi)\begin{bmatrix}{\bm v}_I^{*}-{\bm \eta}_i\\{\bm 0}_n \end{bmatrix}+(\Pi-I_{2n})\begin{bmatrix}{\bm \eta}_i\\{\bm v}_{II}^{*}\end{bmatrix}&=
\Pi\begin{bmatrix}{\bm v}_I^{*}\\{\bm v}_{II}^{*} \end{bmatrix}-\begin{bmatrix}{\bm \eta}_i\\{\bm v}_{II}^{*}\end{bmatrix}\\
&=\begin{bmatrix}{\bm v}_I^{*}-{\bm \eta}_i\\{\bm 0}_n \end{bmatrix}.   
\end{aligned}
$$
Therefore, $\Omega {\bm w}_{n+i}= {\bm w}_{n+i}$ and thus all vectors ${\bm w}_{n+i}$ are eigenvectors of $\Omega$ associated to the eigenvalue in one.

Finally, let us consider  ${\bm w}_{2n+1}=\begin{bmatrix}
    {\bm v}
^{\dag\top} & {\bm 0}_{2n}^\top\end{bmatrix}^\top$, with ${\bm v}
^{\dag}=\begin{bmatrix}{\bm v}
^{\dag\top}_I&
    {\bm v}
^{\dag\top}_{II}
\end{bmatrix}^\top$ the eigenvector associated to the eigenvalue in one of $\Gamma+\Pi$. 
In a compact form, we have that
$$
\begin{bmatrix}{\bm w}_{1},\ldots,{\bm w}_{2n+1}\end{bmatrix}=\begin{bmatrix}
E & {\bm v}_{II}^{*}\otimes {\bm 1}_n^\top -E&{\bm v}
^{\dag}_{I}\\
0_{n\times n}&0_{n\times n}&{\bm v}
^{\dag}_{II}\\
-E & E&{\bm 0}_n\\
0_{n\times n}&{\bm v}_{II}^{*}\otimes {\bm 1}_n^\top&{\bm 0}_n\end{bmatrix}.
$$
Notice that, $E$ is full rank and, being the dominant eigenvector of $C$ and $\Gamma+\Pi$, we have that ${\bm v}_{II}^{*}\neq {\bm 0}_n$ and ${\bm v}
^{\dag}\neq {\bm 0}_{2n}$, respectively. This is sufficient to conclude that the rank of $\begin{bmatrix}{\bm w}_{1},\ldots,{\bm w}_{2n+1}\end{bmatrix}\in\mathbb{R}^{4n \times (2n+1)}$ is $2n+1$. 
In other words ${\bm w}_{1},\ldots,{\bm w}_{2n+1}$ are all linearly independent eigenvectors of $\Omega$ associated to the eigenvalue in one. Hence, the algebraic and geometric multiplicity of the eigenvalue of $\Omega$ in one coincide and are equal to $2n+1$.

At this point, let us consider the matrix $T\in\mathbb{R}^{4n\times 4n}$ that puts $\Omega$ in Jordan Normal Form. In particular, let us consider a matrix $T$ that features ${\bm w}_1,\ldots,{\bm w}_{2n+1}$ as its last $2n+1$ columns, while the first $2n-1$ columns feature vectors in the form $\begin{bmatrix}
    {\bm v}_i^{\ddag\top}&{\bm 0}_n^\top
\end{bmatrix}^\top$, where the vectors ${\bm v}_i\ddag$ are eigenvector of $\Gamma+\Pi$  associated to eigenvalues with magnitude smaller than one, if  the eigenvalue has coinciding algebraic and geometric multiplicities, or generalized eigenvectors, otherwise.
We have that
$$
\Omega=T\begin{bmatrix}
        \widetilde{\Omega}&0\\
        0&I_{2n+1}
    \end{bmatrix}T^{-1}
$$
with $\widetilde{\Omega}\in\mathbb{R}^{(2n-1)\times (2n-1)}$ having just eigenvalues with magnitude strictly smaller than one (i.e., the eigenvalues in $\Gamma+\Pi$, without the eigenvalue in one). Notice that, since the geometric and algebraic multiplicity of the eigenvalue in one coincide, the Jordan normal form features $2n+1$ scalar diagonal blocks equal to one, i.e., it features a diagonal block $I_{2n+1}$. 
As a consequence, we have that
\begin{align*}
 \lim_{k\to \infty}\Omega^k &= T\begin{bmatrix}
        \lim_{k\to \infty} \widetilde{\Omega}^k&0\\
        0&I_{2n+1}
    \end{bmatrix}T^{-1} \\[0.2cm]
    &= T\begin{bmatrix}
        0&0\\
        0&I_{2n+1}
    \end{bmatrix}T^{-1}   ,
\end{align*}
and thus ${\bm y}_k=\Omega^k{\bm y}_0$ converges to a steady state, and so does ${\bm z}_k$, and the geometric multiplicity of the eigenvalue in one is at least $n$.


Finally, we consider the case for which ${\bm e}_k = \check{\bm z}_k - {\bm z}_k~\forall k\in\mathbb{Z}_{\geq 0}$. From \eqref{eq:quantizer}, we can infer that by properly adjusting the quantizers' midpoint $\sigma_k$ and step size $\Delta_k$, then by the construction of the quantizers we can always keep ${\bm e}_k$ bounded within $\Delta_k/2$. For a bounded error, iteration still converges according to the discussion in the paragraph above. Nevertheless, by having $\Delta_k$ decreasing according to the coordination variable $\zeta_k$, and by shifting the midpoint $\sigma_k$ as in \eqref{eq:shifting} to track the actual average consensus value, then ${\bm e}_k$ is driven towards $0$ as $k\rightarrow\infty$. Then by the preservation of the total mass in the network at all times $k\in\mathbb{Z}_{\geq 0}$ from Lemma~\ref{lemma:1} 
we have that the iteration in \eqref{eq:aug_form_error} converges to ${{\bm z}^\star} = [\mathbf{1}^{\top}x_{\text{ave}}, \mathbf{0}_n^{\top}]^{\top}$ as $k\rightarrow\infty$.

\bibliographystyle{IEEEtran}
\bibliography{references}

\end{document}